\documentclass[aoas,preprint]{imsart}

\def\bSig\mathbf{\Sigma}

\newcommand{\D}{{\cal{D}}}
\newcommand{\HS}{\mathrm{HS}}

\usepackage[figuresright]{rotating}
\usepackage{graphicx}
\usepackage{subfigure}
\usepackage{amsmath}
\usepackage{hyperref}
\usepackage{multirow}
\usepackage{algorithm}
\usepackage{algorithmic}
\usepackage{enumerate}
\usepackage{csvsimple}
\usepackage{datatool}
\usepackage{booktabs}
\usepackage{xcolor}
\usepackage{amsfonts}
\usepackage[normalem]{ulem}

\makeatletter

\newcommand{\Rmnum}[1]{\expandafter\@slowromancap\romannumeral #1@}
\makeatother
\RequirePackage[OT1]{fontenc}
\RequirePackage{amsthm,amsmath}
\RequirePackage[authoryear]{natbib}
\RequirePackage[colorlinks,citecolor=blue,urlcolor=blue]{hyperref}

\arxiv{arXiv:1707.03301}

\startlocaldefs
\numberwithin{equation}{section}
\theoremstyle{plain}
\newtheorem{thm}{Theorem}[section]
\endlocaldefs

\begin{document}

\begin{frontmatter}
\title{
Bayesian latent hierarchical model for transcriptomic meta-analysis to detect biomarkers with
clustered meta-patterns of differential expression signals
\thanksref{T1}
}
\runtitle{Bayesian meta-analysis for biomarkers of meta-patterns}
%\runtitle{BayesMP}

\begin{aug}
\author[1]{\fnms{Zhiguang} \snm{Huo}\ead[label=e1]{zhuo@ufl.edu}},
\author[2]{\fnms{Chi} \snm{Song}\ead[label=e2]{song.1188@osu.edu}}$^,$\thanksref{co},
\and
\author[3]{\fnms{George} \snm{Tseng}\ead[label=e3]{ctseng@pitt.edu}}$^,$\thanksref{co}

\thankstext{co}{To whom correspondence should be addressed.}
\thankstext{T1}{Supported by NIH RO1CA190766 for ZH and GT.}

\runauthor{Z. Huo et al.}

\affiliation[1]{University of Florida}
\affiliation[2]{The Ohio State University} \and 
\affiliation[3]{University of Pittsburgh}

\address{Zhiguang Huo\\
Department of Biostatistics\\ 
University of Florida \\
Gainesville, FL 32611\\
\printead{e1}\\
\phantom{E-mail:\ }}

\address{Song Chi\\
Division of Biostatistics\\
College of Public Health\\ 
The Ohio State University\\ 
Columbus, OH 43210\\
\printead{e2}\\
}

\address{George Tseng\\
Department of Biostatistics, Human Genetics \\ 
and Computational Biology\\
University of Pittsburgh \\
Pittsburgh, PA 15261\\
\printead{e3}\\
}
\end{aug}

\begin{abstract}
Due to the rapid development of high-throughput experimental techniques and fast-dropping prices,
many transcriptomic datasets have been generated and 
accumulated in the public domain.
Meta-analysis combining multiple transcriptomic studies can increase the statistical power to detect disease-related biomarkers.
In this paper, we introduce a Bayesian latent hierarchical model 
to perform transcriptomic meta-analysis.
This method is capable of detecting genes that are differentially expressed (DE) in only a subset of the combined studies,
and the latent variables help quantify homogeneous and heterogeneous differential expression signals across studies.
A tight clustering algorithm is applied to detected biomarkers to capture differential meta-patterns that are informative to guide further biological investigation.
Simulations and three examples, including a microarray dataset from metabolism-related knockout mice, 
an RNA-seq dataset from HIV transgenic rats, 
and cross-platform datasets from human  breast cancer,
are used to demonstrate the performance of the proposed method.

\end{abstract}


\end{frontmatter}

\section{Introduction}
\label{s:intro}

With the rapid development of high-throughput experimental techniques and fast-dropping prices,
many transcriptomic datasets have been generated and deposited into public databases.
In general, each dataset contains a small to moderate sample size,
which requires caution in gauging the accuracy and reproducibility of detected biomarkers \citep{simon2003pitfalls,simon2005development,domany2014using}.
Meta-analysis combining multiple transcriptomic studies can increase statistical power and provide robust conclusions from various platforms and sample cohorts \citep{ramasamy2008key}.
\citet{tseng2012comprehensive} presented a comprehensive review of methods and applications in the microarray meta-analysis field and categorized existing methods into combining \textit{p}-values,
combining effect sizes, direct merging (aka mega-analysis),
and combining nonparametric ranks.
In general, meta-analysis can be viewed as two-step information reduction and combination tools for adjusting batch effects
(such as different experimental platforms, protocols, and bias)
across studies to draw a more efficient and accurate conclusion.
This paper focuses on combining \textit{p}-value methods, 
and we will discuss other potential approaches of adjusting batch effects and directly merging studies in the discussion section.

Following conventions in \citet{birnbaum1954combining} and \citet{li2011adaptively},
two hypothesis settings have been considered in meta-analysis.
In the first setting (namely $\HS_A$),
we aim to detect biomarkers that are differentially expressed (DE) in all studies: $H_0: \vec{\theta}\in\bigcap \{ \theta_s=0\}$
versus $H_A: \vec{\theta}\in\bigcap \{ \theta_s \ne 0\}$,
where $\theta_s$ is the effect size of study $s$, 
$1\le s \le S$.
Throughout this manuscript,
effect size refers to unstandardized effect size \citep{cooper2009handbook} (difference of group means or unstandardized regression coefficients).  
In the second setting ($\HS_B$), 
targeted biomarkers are DE in one or more studies:
$H_0: \vec{\theta}\in\bigcap \{ \theta_s=0 \}$ versus 
$H_A: \vec{\theta}\in\bigcup \{ \theta_s \ne 0 \}$.
In view of overly stringent criterion in $\HS_A$ in noisy genomic data and when $S$ is large,
\citet{song2014hypothesis} proposed a robust setting $\HS_r$, requiring that $r$ or more studies are 
DE: $H_0: \vec{\theta}\in\bigcap \{ \theta_s=0 \}$ versus 
$H_A: \vec{\theta}\in\sum \mathbb{I} { \{ \theta_s \ne 0 \} } \ge r$,
where $\mathbb{I} { \{ \cdot \} } $ is an indicator function taking value one if the statement is true and zero otherwise, 
and $r$ is usually pre-estimated with $S/2 \le r \le S$. \citet{song2014hypothesis} also proposed using the $r$th-ordered \textit{p}-value (rOP, $T^{rOP}=p_{(r)}$) to test $\HS_r$.
Generally speaking, $\HS_A$ and $\HS_r$ are most biologically interesting, because they are designed to  detect consensus biomarkers across the combined studies.
However, when heterogeneous differential expression signals across studies are expected 
(e.g., studies come from different tissues or brain regions in the two mouse/rat examples in section~\ref{sec:realData}),
biomarkers detected from $\HS_B$ can also be of interest.
\citet{chang2013meta} conducted a comparative study evaluating twelve popular microarray meta-analysis methods targeting  the three hypothesis settings ($\HS_A$, $\HS_B$, and $\HS_r$).

Strictly speaking,
$\HS_B$ is a sound {\color{black} hypothesis setting}, and statistical tests for $\HS_B$ are easier to develop.
Most popular \textit{p}-value aggregation methods,
such as Fisher's method \citep{fisher1934statistical} and Stouffer's method \citep{stouffer1949american}, 
aim for this $\HS_B$ setting.
In the literature, $\HS_B$ is also called a conjunction or intersection hypothesis \citep{benjamini2008screening}.
On the other hand, 
$\HS_A$ is a {\color{black}somewhat defective hypothesis setting in the sense that the null and alternative spaces are not complementary}.
For example, if we apply the maximum \textit{p}-value test ($T^{maxP} = \max_s p_s$,
where $p_s$ is the \textit{p}-value for study $s$) for $\HS_A$, 
and we reasonably assume that \textit{p}-values independently follow UNIF(0,1) under the null hypothesis,
then the null distribution of $T^{maxP}$ is Beta(S,1). This test is anticonservative 
when there exist genes DE in some but not all of the studies.
The problem mainly comes from the noncomplementary null and alternative spaces in $\HS_A$ 
(and also $\HS_r$).
A more appropriate {\color{black} hypothesis setting for the same purpose} is 
$\HS_{\bar{A}} \equiv H_0: \vec{\theta}\in\bigcup \{ \theta_s=0 \}$ versus 
$H_A: \vec{\theta}\in\bigcap \{ \theta_s \ne 0 \}$ 
(and $\HS_{\bar{r}} \equiv H_0: \sum \mathbb{I} \{ \theta_s \ne 0 \} < r$
versus $H_A: \sum \mathbb{I} \{ \theta_s \ne 0 \} \ge r$).
\citet{benjamini2008screening} proposed a legitimate but conservative test for $\HS_{\bar{A}}$ and $\HS_{\bar{r}}$.
In genomic applications, 
the composite null hypotheses of $\HS_{\bar{A}}$ and $\HS_{\bar{r}}$ are complicated 
by the fact that genes can be differentially expressed in up to $S-1$ (for $\HS_{\bar{A}}$) or $r-1$ (for $\HS_{\bar{r}}$) studies, 
with different levels of effect sizes.
Under such a scenario, 
it becomes very difficult to characterize the null distribution for hypothesis tests in a frequentist setting.
Theoretically it is necessary to borrow differential expression information across genes to significantly improve statistical power.
Bayesian hierarchical modeling can provide a convenient solution for this purpose.
\citet{efron2008microarrays} and \citet{efron2009empirical}  applied empirical Bayes methods to control the false discovery rate (FDR) in single microarray studies.  
\citet{muralidharan2010empirical} further extended these works to allow for the simultaneous modeling of both empirical null and alternative distribution of the test statistics. 
 \citet{zhao2014bayesian} incorporated pathway information to select genes using the Bayesian mixture model.  
Despite these successful applications, a Bayesian method that combines multiple studies and detects DE genes based on various meta-analysis hypothesis settings is yet to be developed.
In this paper, we propose a Bayesian latent hierarchical model 
(BayesMP, named after the Bayesian method for meta-patterns) that uses a nonparametric Bayesian method to effectively combine information across genes 
for direct testing for $\HS_{\bar{A}}$ 
(as well as $\HS_B$ and $\HS_{\bar{r}}$) on a genome-wide scale.
In simulations, we show successful Bayesian false discovery rate control of BayesMP, 
while the original maxP and rOP method using a beta null distribution loses FDR control for the $\HS_{\bar{A}}$ and $\HS_{\bar{r}}$ settings.

Traditionally, meta-analysis aims to pool consensus signals to increase statistical power (e.g., by fixed or random effects models).
Recently, researchers have recognized the existence of heterogeneous signals among cohorts
and the importance of their characterization in meta-analysis.
For example, figure~\ref{fig:mouse2_modulesA} shows three modules of detected biomarkers from
the RNA-seq HIV transgenic rat data using BayesMP and meta-pattern clustering (see section~\ref{s:mouseHIV} for details).
Modules  \Rmnum{1} and  \Rmnum{2} are consensus biomarkers that are either all down-regulated or
all up-regulated across three brain regions.
In contrast, the biomarkers in module \Rmnum{3} are down-regulated in HIP but up-regulated in PFC and STR.
Such  biomarkers are somewhat expected because it is well known that different brain regions are responsible for different functions such as reasoning, recognition, visual inspection, and memory/speech. 
Several approaches, such as the adaptive weighting (or subset) method \citep{li2011adaptively, bhattacharjee2012subset}
and lasso variable selection \citep{li2014meta},
have been proposed for quantifying and inferring such heterogeneity.
In the adaptively weighted Fisher's method (AW-Fisher; \Citealt{li2011adaptively}), 
for example,
heterogeneity of differential expression signals in each study is categorized by $w_{gs}$ as 0 or 1 weights
(1 representing differential expression for gene $g$ in study $s$, 
and 0 for nondifferential expression).
Specifically, AW-Fisher considers weighted Fisher's statistics 
--- that is, $T(\vec{W}_g) = -2\sum_{s=1}^S w_{gs} \cdot \log(p_{gs})$,
where $\vec{W}_g = (w_{g1},\ldots, w_{gS})$ is the vector of 0 or 1 weights reflecting gene-specific heterogeneous contribution of each study and 
$p_{gs}$ is the \textit{p}-value of gene $g$ in study $s$ ---
and adaptively searches the best weight vector for gene $g$ by minimizing the resulting \textit{p}-value:
$\widehat{\vec{W}_g} = \arg\min_{\vec{W}_g}
p(T({\vec{W}}_g)) = \arg\min_{\vec{W}_g} 1 - F^{-1}_{\chi^2_{2\cdot\sum_{s=1}^S w_{gs}}} (T(\vec{W}_g))$, 
where $F^{-1}_{\chi^2_{w}}$ is the inverse CDF of chi-squared distribution with degree of freedom $w$.
The 0 or 1 weights estimated from AW-Fisher help cluster detected biomarkers by their differential expression meta-patterns 
but have a disadvantage of hard-thresholding without quantification of variability.
When $S$ is large, 
the number of all possible $2^S - 1 $  weight combinations also makes the problem intractable.
In BayesMP, the differential expression indicators naturally come with variability estimates from posterior distribution
(see a confidence score to be defined later in section~\ref{sec:decision}).
In BayesMP, we also adopt a cosine dissimilarity measure on these posterior distributions and apply tight clustering \citep{tseng2005tight} to identify biomarkers of different meta-patterns
(e.g., see the three modules of biomarkers in figure~\ref{fig:mouse2_modules}).
Unsupervised clustering of the expression pattern across studies identifies data-driven modules of biomarkers of different meta-patterns and provides 
interpretable results for further biological investigation.
For example,
it is interesting to investigate why biomarkers in module \Rmnum{3} are down-regulated in HIP but up-regulated in PFC and STR.
We note here that our proposed cluster analysis to categorize detected biomarkers by studying heterogeneity in meta-analysis is a relatively novel concept.
It is different from popular practices of clustering genes for identifying coexpression gene modules or clustering samples for discovering disease subtypes (e.g., \cite{huo2016meta}).

\begin{figure}[htbp]
	\label{fig:mouse2_modules}
	\centering
	\subfigure[Heatmap]{
		\label{fig:mouse2_modulesA}
		\includegraphics[height=0.6\columnwidth]{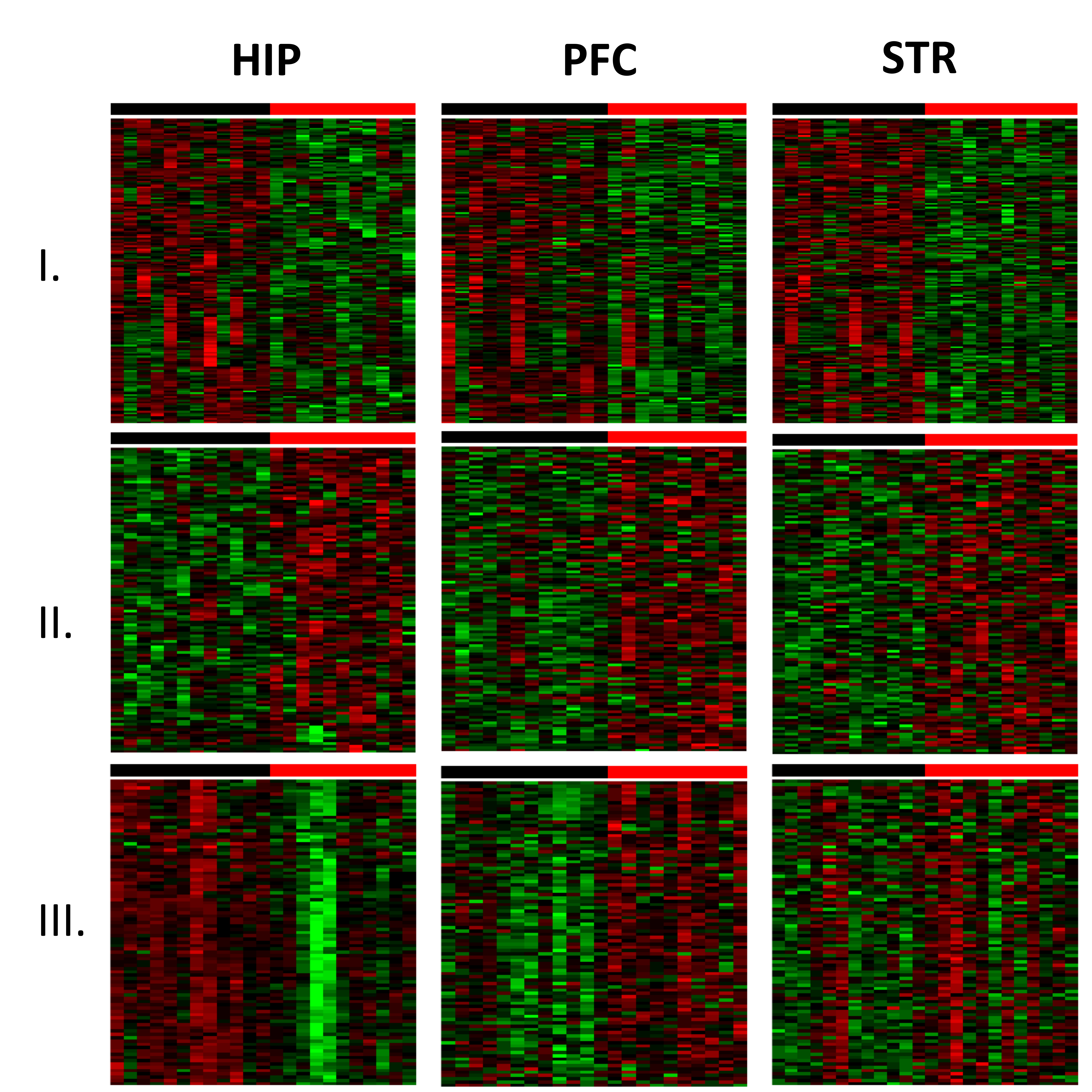}
	}
	\subfigure[CS]{
		\label{fig:mouse2_modulesB}
		\includegraphics[height=0.6\columnwidth]{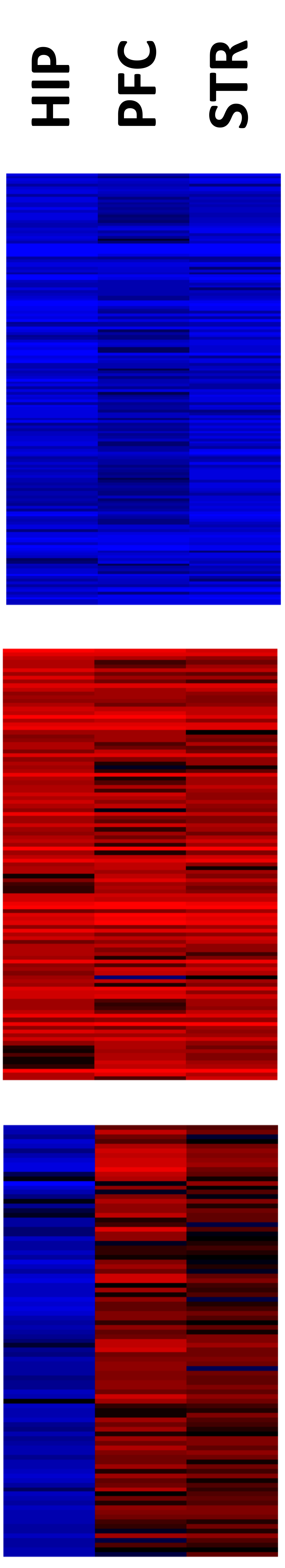}
	}
	\subfigure[bar plot]{
		\label{fig:mouse2_modulesC}
		\includegraphics[height=0.6\columnwidth]{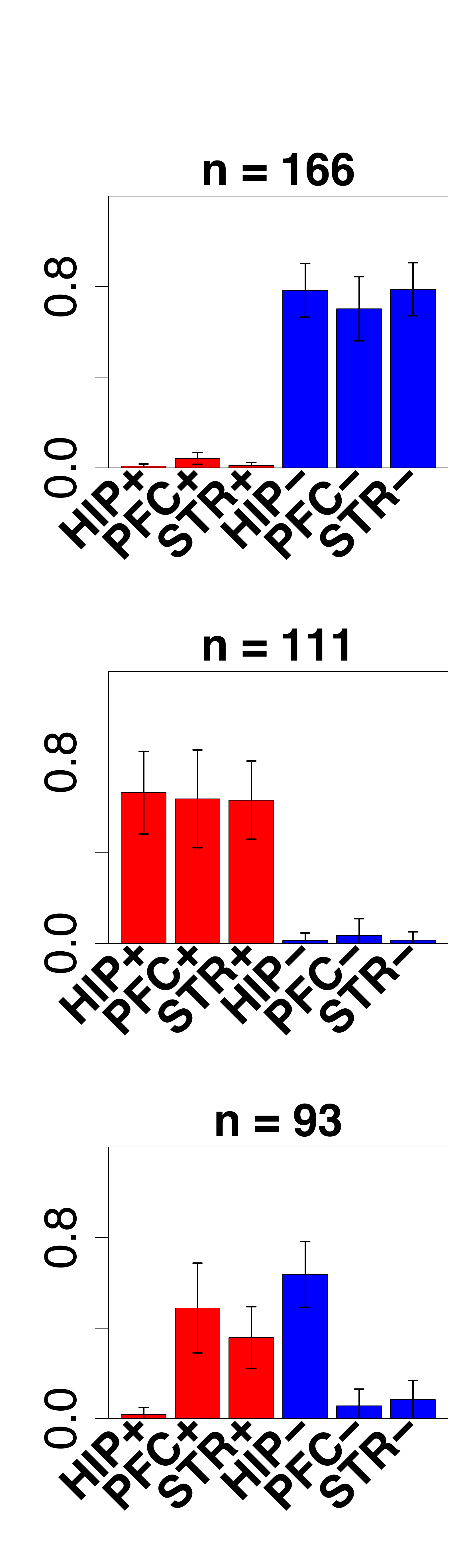}
	}
	\caption{
	Three meta-pattern modules of biomarkers from HIV transgenic rats example.
	Each row (module \Rmnum{1}, \Rmnum{2}, and \Rmnum{3}) shows a set of detected biomarkers showing similar meta-pattern of differential signals.
	\ref{fig:mouse2_modulesA} Heatmaps of detected genes (on the rows) and samples (on the columns) for each brain region (HIP, PFC, or STR), where each brain region represents a study (i.e., HIP for $s=1$, PFC for $s=2$, STR for $s=3$).	
	The black color bar on top represents F334 rats (control), 
	and the red color bar on top represents HIV transgenic rats (case).
	\ref{fig:mouse2_modulesB} Heatmaps of confidence scores (CS) (genes on the rows and three studies on the columns). 
	The confidence score is described in section~\ref{sec:decision}, which ranges from -1 (blue color for down-regulation) to 1 (red color for up-regulation).
	\ref{fig:mouse2_modulesC} Bar plots of mean posterior probability for $Y_{gs} =1$ (red color for up-regulation) and $Y_{gs} = -1$ (blue color for down-regulation) for each module in each brain region.
	Error bar represents standard deviation across all genes in the module.
	The number of genes is shown on top of each bar plot.
	}
\end{figure}

In this paper, 
section~\ref{s:methods} establishes the methodology, estimation, and inference of BayesMP.
Section~\ref{sec:simu} evaluates the performance of the proposed method using simulation datasets.
Section~\ref{sec:realData} shows the application to three real examples.
Finally, section~\ref{sec:conclusion} provides conclusions and discussions.

\section{Methods}
\label{s:methods}

For the ease of discussion, we focus on detecting DE genes in two-class comparison in this manuscript.
The method can be easily extended for studies with numerical or survival outcomes.
In a meta-analysis combining $S$ studies with $G$ genes, we denote $p_{gs}$ as the one-sided \textit{p}-value 
testing for down-regulation for gene $g$ in study $s$, where $1\le g \le G$ and $1\le s\le S$.
These \textit{p}-values can be calculated from SAM \citep{tusher2001significance} 
or \textit{limma} \citep{smyth2005limma} for microarray studies (or RNA-seq studies with RPKM data), and \textit{edgeR} \citep{robinson2010edger} or \textit{DEseq} \citep{anders2010differential}  for RNA-seq studies with count data.
As a result, our model is flexible to mixed studies of different platforms (e.g., microarray or RNA-seq) and study designs (e.g., case-control, numerical outcome, or survival outcome). 
Throughout this manuscript, we use \textit{limma} and \textit{edgeR} to obtain the \textit{p}-values.
For modeling convenience, we transform the one-sided \textit{p}-values into \textit{Z}-statistics, i.e.
$Z_{gs} = \Phi^{-1}(p_{gs})$,
where $\Phi^{-1}(\cdot)$ is the inverse cumulative density function (CDF) of standard Gaussian distribution.
$Z_{gs}$ is the input data for BayesMP.

\subsection{Bayesian hierarchical mixture model} 
\label{sec:BayesianModel}
Denote by $\theta_{gs}$ the effect size of gene $g$ in study $s$ and by $Y_{gs}$ an indicator variable s.t.
$Y_{gs}=1$ if $\theta_{gs}>0$ (up-regulation),
$Y_{gs}=-1$ if $\theta_{gs}<0$ (down-regulation),
and $Y_{gs}=0$ if $\theta_{gs}=0$ (non-DE gene).
We assume that the \textit{Z}-statistics from study $s$ are sampled from a mixture distribution with three mixing components depending on $Y_{gs}$: 
 $f^{(s)}(Z_{gs}|Y_{gs})= f_0^{(s)}(Z_{gs})\cdot \mathbb{I}(Y_{gs}=0) + f_{+1}^{(s)}(Z_{gs})\cdot \mathbb{I}(Y_{gs}=1) + f_{-1}^{(s)}(Z_{gs})\cdot \mathbb{I}(Y_{gs}=-1)$,
where $f^{(s)}(\cdot)$ is the pdf of \textit{Z}-statistics in study $s$, and $f_0^{(s)}$, $f_{+1}^{(s)}$ and $f_{-1}^{(s)}$ are the pdfs of the null, positive, and negative components in study $s$.

 In most situations, if an appropriate statistical test is adopted, one can expect that $p_{gs}\sim \textrm{Unif}(0,1)$
if gene $g$ in study $s$ is not DE and hence reasonably assume $f_0^{(s)}\equiv \textrm{N}(0,1)$.  
If the \textit{p}-value distribution is not uniform under null hypothesis, one can also empirically estimate $f_0^{(s)}$ following \citet{efron2004large}.  
Throughout this manuscript, we use theoretical null $\textrm{N}(0,1)$, 
and we put discussion about this choice in the conclusion section.
Unlike $f_0^{(s)}$, $f_{\pm1}^{(s)}$ are usually unknown, 
and their estimation is not trivial.  
To account for the complex composition of alternative $f_{\pm1}^{(s)}$ (potentially several subgroups exist in the alternative space), we model them nonparametrically by assuming they are also mixtures of distributions using Dirichlet processes (DPs). 
DPs are widely discussed and applied in the literature \citep{neal2000markov, muller2004nonparametric},
and density estimation using DPs has also been discussed \citep{escobar1995bayesian}.  
In our model, when $Y_{gs}\ne 0$, we assume $Z_{gs}\sim \mbox{N}(\mu_{gs},1)$, and $\mu_{gs}$ follows distribution $G_{s+}$ or $G_{s-}$ generated from DPs.  Specifically, the generative process of $Z_{gs}$ given $Y_{gs}=\pm 1$ is as follows:

\begin{equation*}
\centering
\begin{split}
& G_{s+}\sim \textrm{DP}(G_{0+}, \alpha_{+}) \mbox{ and } G_{s-}\sim \textrm{DP}(G_{0-}, \alpha_{-}).\\
& \mu_{gs}\sim\left\{\begin{split}G_{s+} & \mbox{ if } Y_{gs}=1,\\
G_{s-} & \mbox{ if } Y_{gs}=-1.\end{split}\right.\\
& Z_{gs}\sim \mbox{N}(\mu_{gs}, 1).
\end{split}
\end{equation*}

$\textrm{DP}(G, \alpha)$ denotes a Dirichlet process with base distribution $G$ and concentration parameter $\alpha$, 
and $G_{0+}$ ($G_{0-}$) denotes normal density $\mbox{N}(0,\sigma_0^2)$ left (right) truncated at 0.  
We find that the selection of $\sigma_0$ and $\alpha_{\pm}$ do not much affect the performance of this model in simulation [see details in supplementary table~S3 \citep{supp}].  
It should be noted that we assume $Z_{gs}\sim \mbox{N}(\mu_{gs}, 1)$, 
where the variance is fixed at 1 to ensure that $f^{(s)}_{+1}/f^{(s)}_0$ ($f^{(s)}_{-1}/f^{(s)}_0$) is monotonically increasing (decreasing) while $Z_{gs}>0$ ($Z_{gs}<0$), 
which in turn guarantees the posterior probability of gene $g$ being DE in study $s$ to increase as $|Z_{gs}|$ increases (see Theorem~\ref{thm:monotone1}).  
In addition, this assumption makes the MCMC simpler and hence speeds up the algorithm.

\begin{thm}
\label{thm:monotone1}
If $f_k(x)\equiv\mbox{N}(\mu_k,\sigma_k^2)$ with $\mu_k>0$, $\sigma_k^2 \ge 1$ and $1\le k\le K$,
 and $f_0(x)=\mbox{N}(0,1)$,
then $\sum_{k=1}^K w_kf_k/f_0$ is monotonically increasing when $x\ge0$, 
where $w_1, \ldots, w_K>0$ and 
$\sum_{k=1}^K w_k=1$.
\end{thm}
\begin{proof}
$\forall 1\le k\le K$, $f_k(x)=1/(\sigma_k\sqrt{2\pi}) \exp(-(x-\mu_k)^2/(2\sigma_k^2))$,  
$f_0(x) = 1/\sqrt{2\pi} \exp(-x^2/2)$,
So $g_k(x)=f_k(x)/f_0(x)=1/\sigma_k \exp(-(x-\mu_k)^2/(2\sigma_k^2))+x^2/2)$
and $\log g_k(x) = -(x-\mu_k)^2/(2\sigma_k^2) + x^2/2 - \log \sigma_k$.
By taking derivative of $\log g_k(x)$, we get 
$[\log g_k(x)]'=-(x-\mu_k)/\sigma_k^2 + x = (1 - 1/\sigma_k^2)x + \mu_k/\sigma_k^2 > 0$,
when $x \ge 0$ (actually $x > \mu_k/(1-\sigma_k^2)$ is enough).
Therefore $g_k=f_k/f_0$ is monotonically increasing when $x \ge 0$, and $\sum_{k=1}^K w_kf_k/f_0$ is also monotonically increasing when $x \ge 0$.
\end{proof}

In order to borrow information across studies, we further assume that $Y_{gs}$ is generated depending on 
(1) the prior probability $\pi_g$ that gene $g$ is a DE gene and 
(2) the conditional probability $\delta_g$ for gene $g$
in study $s$ being up-regulated (or $1 - \delta_g$ for down-regulated),
given gene $g$ is DE.
Specifically, we assume $\vec{W}_{gs}~\sim~\mbox{Mult}\left(1,
(1 - \pi_g,
\pi_g^+,
\pi_g^-)
\right)
$
and 
$Y_{gs} =  \vec{W}_{gs}
 \cdot 
(0,1,-1)
$,
where $\pi_g^+ = \pi_g \delta_g$, $\pi_g^- = \pi_g (1 - \delta_g)$ and $\cdot$ is the inner product of two vectors.
Given $Y_{gs}=y$, $Z_{gs}$ is generated from $f_y^{(s)}(Z)$.
The graphical representation of the full generative model is shown in figure~\ref{fig:graphicalModel}. 

We assume that each gene $g$ is DE in different studies in the same probability $\pi_g$, i.e., $\pi_g=\Pr(Y_{gs}\ne0)$, 
and $\pi_g \sim \mbox{Beta}(\gamma,1 - \gamma)$.  
$\gamma$ can be interpreted as the proportion of DE genes pooling all studies, 
since the expectation of $\pi_g$ from this prior is $\gamma$.
We further set the prior of $\gamma$ being uniform distribution ($\gamma \sim \mbox{UNIF}(0,1)$).

For each gene $g$, define $\delta_g = \Pr(Y_{gs}=1|\mbox{gene } g\mbox{ is a DE gene})$.
We assume $\delta_g \sim \mbox{Beta}(\beta, \beta)$.
We set $\beta = 1/2$ in this paper, which gives a noninformative prior.
Note that this conditional probability provides flexibility for a DE gene to contain conflicting differential expression directions 
(i.e., up-regulation in one study but down-regulation in another study; 
e.g., module \Rmnum{3} in figure~\ref{fig:mouse2_modules})

\begin{figure}[htbp]
	\centering
	\includegraphics[width=0.6\columnwidth]{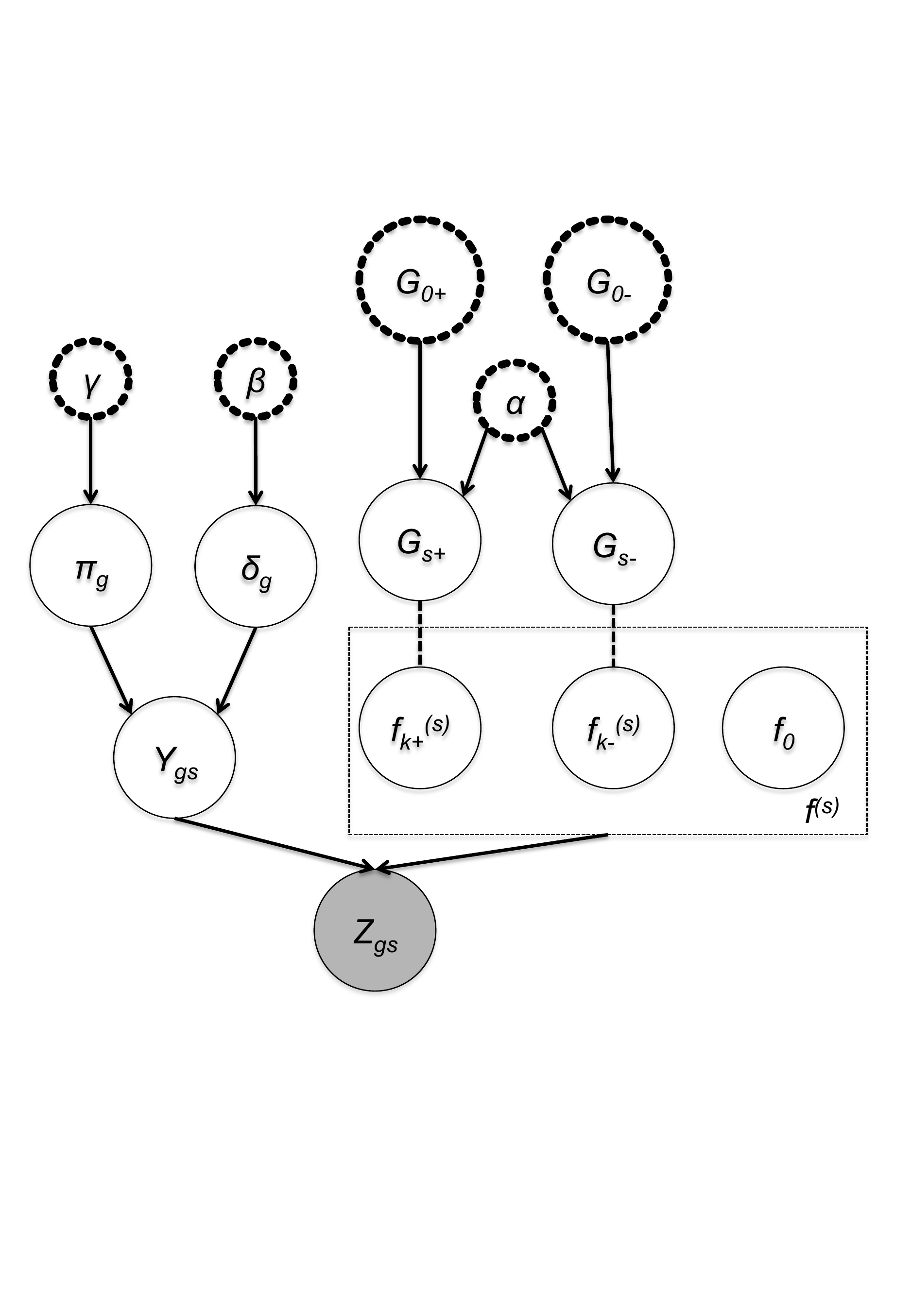}
	\caption{Graphical representation of Bayesian latent hierarchical model.
	Shaded nodes are observed variables.
	Dashed nodes are pre-estimated/fixed parameters.
	Arrows represent generative process.
	Dashed lines represent equivalent variables.
	$s$ is the study index, and $g$ is the gene index.  }
	\label{fig:graphicalModel}
\end{figure}

\subsection{Model fitting}

Since conjugate priors were used in the generative model, 
we can generate  the posterior samples efficiently using the MCMC procedure.
In order to update the DP with an infinite number of components, we take the alternative view of DP as the Chinese restaurant process. 
We define $C_{gs}\in \{\ldots, -3, -2, -1, 0, 1, 2, 3, \ldots\}$ as the auxiliary component variable, and $Y_{gs}$ is determined by the sign of $C_{gs}$.
Specifically,  if $C_{gs}=0$, we set $Y_{gs}=0$;
if $C_{gs}=k$ with $k>0$, then we set $Y_{gs}=1$, 
and sample $\mu_{gs}$ from the $k$th component (or $k$th table in Chinese restaurant process) of $G_{s+}$; 
and similarly, if $C_{gs}=k$ with $k<0$, then we set $Y_{gs}=-1$ and sample $\mu_{gs}$  from the $k$th component of $G_{s-}$.
The following steps provide details of our MCMC iterations:
\begin{enumerate}
\item Update $\pi_g$'s:
$$\pi_g | Y_{gs} \sim \mbox{Beta}(\gamma + Y_g^+ + Y_g^-, S - Y_g^+ - Y_g^-  + 1 - \gamma),$$
where $Y_g^+ = \sum_s \mathbb{I}(Y_{gs}=1)$ and $Y_g^- = \sum_s \mathbb{I}(Y_{gs}=-1)$.
\item Update $\delta_g$'s:
$$\delta_g| Y_{gs} \sim \mbox{Beta}(\beta + Y_g^+, \beta + Y_g^-).$$
\item Update $Y_{gs}$'s: 

First update 
$C_{gs}$'s s.t.
$$	\Pr(C_{gs}=k | C_{-g,s}, Z_{gs},  \pi_g^{\pm})  \propto h_k^{(s)}(Z_{gs} | C_{-g,s}) (\pi_g^+)^{\mathbb{I}(k>0)} (\pi_g^-)^{\mathbb{I}(k<0)} (1 - \pi_g)^{\mathbb{I}(k=0)},
$$
where $C_{-g,s}$ denotes all the $C$'s in study $s$ excluding gene $g$.
Note that $h_k^{(s)}$ can be calculated directly 
following the convention of algorithm 3 in \cite{neal2000markov}.
Details of $h_k^{(s)}(Z_{gs} | C_{-g,s})$ are given in supplementary~section~\Rmnum{3} \citep{supp}.
Finally, we set $Y_{gs} = \mbox{sgn}({C_{gs}})$,
where $\mbox{sgn}(\cdot)$ is the sign function.

\item 
Because there is no immediate conjugate prior for $\gamma$,
we sample $\gamma$ using the Metropolis-Hastings algorithm such that
$$\gamma \propto \prod_{g=1}^G \mbox{dBeta}(\pi_g; \gamma, 1 - \gamma),$$
where $\mbox{dBeta}(x; a, b)$ is the probability density function of beta distribution 
evaluated at $x$, with the shape parameters $a$ and $b$.
For details see supplementary section~\Rmnum{1} \citep{supp}.

\end{enumerate}

\subsection{Decision space and inference making}
\label{sec:decision}
A main benefit of Bayesian modeling is its capability of making inference by statistical decision theory,
a generalized framework that covers the traditional hypothesis testing framework as a special case \citep{berger2013statistical}.
Take $\HS_{\bar{A}}$ from section~\ref{s:intro} as an example.
Traditional hypothesis testing considers null hypothesis $H_0: \vec{\theta}_g \in \Omega_{\bar{A}}^0 $, 
where $\Omega_{\bar{A}}^0 =\{\vec{\theta}_g:  \sum_{s=1}^S\mathbb{I}(\theta_{gs}\ne 0)<S \}$ and
$\vec{\theta}_g =(\theta_{g1}, \cdots, \theta_{gS})$, 
versus alternative hypothesis $H_A: \vec{\theta}_g\in\Omega_{\bar{A}}^1$
where $\Omega_{\bar{A}}^1 =\{\vec{\theta}_g:  \sum_{s=1}^S\mathbb{I}(\theta_{gs}\ne 0)=S \}$.
When observed data are unlikely to happen 
(i.e., type \Rmnum{1} error controlled at $5\%$) under null hypothesis,
we reject the null hypothesis.
One notable feature is that traditional hypothesis testing views null and alternative hypothesis spaces differently. 
The decision of hypothesis testing is only based on the null hypothesis --- either to reject or to accept. 
The alternative hypothesis plays little role in decision making.
In view of decision theory framework, 
a decision space (aka action space) is designed as $\D_{\bar{A}}=(\Omega_{\bar{A}}^0, \Omega_{\bar{A}}^1)$. 
The inference generates a decision function $f$ that maps from the observed data space $Z$ to $\D_{\bar{A}}$ 
(i.e., $f:Z\rightarrow \D_{\bar{A}}$). 
Under this framework,
the type \Rmnum{1} error can be expressed as $\Pr(f(Z) = \Omega_{\bar{A}}^1| \vec{\theta}_g \in \Omega_{\bar{A}}^0)$,
and the type \Rmnum{2} error as $\Pr(f(Z) = \Omega_{\bar{A}}^0| \vec{\theta}_g \in \Omega_{\bar{A}}^1)$.
Hypothesis testing is a special case under this framework, 
with adequate type \Rmnum{1} error
control to determine the decision function $f$.
Unlike hypothesis testing,
decision theory treats $\Omega_{\bar{A}}^0$ and $\Omega_{\bar{A}}^1$ equally, because in decision theory, the decision
is made through cost analysis, which weighs the costs of making wrong decisions in both spaces. 
One can easily design a realistic loss (cost) function based on the two types of errors to determine their balance and achieve the best decision function.
In this paper, in order to make a fair comparison with classical hypothesis testing,
we use posterior probabilities from Bayesian modeling and adopt a false discovery control described by
\cite{newton2004detecting} to determine the decision function.

Denote by $\xi_g=\Pr(\vec{\theta}_g\in \Omega_{\bar{A}}^0|Z)= 1-\Pr(\vec{\theta}_g\in \Omega_{\bar{A}}^1|Z)$,
which is local FDR \citep{efron2002empirical} by definition.
Given a threshold $\kappa$,
we declare gene $g$ as a DE gene if $\xi_g \le \kappa$ and the expected number of false discoveries is 
$\sum_g \xi_g \mathbb{I} (\xi_g \le \kappa)$.
The false discovery rate from Bayesian modeling is defined as $\frac{\sum_g \xi_g \mathbb{I} (\xi_g \le \kappa)}{ \sum_g \mathbb{I} (\xi_g \le \kappa)}$ \citep{newton2004detecting}.
In simulation and real data applications,
we compare the performance of FDR control from traditional hypothesis testing and FDR control from Bayesian modeling.
Note that we can consider $\D_{B}=(\Omega_B^0, \Omega_B^1)$, where 
$\Omega_B^0=\{\vec{\theta}_g:  \sum_{s=1}^S\mathbb{I}(\theta_{gs}\ne 0)=0 \}$ and 
$\Omega_B^1=\{\vec{\theta}_g:  \sum_{s=1}^S\mathbb{I}(\theta_{gs} \ne 0)>0\}$,
to correspond to $\HS_B$; and 
$\D_{\bar{r}}=(\Omega_{\bar{r}}^0, \Omega_{\bar{r}}^1)$, where 
$\Omega_{\bar{r}}^0=\{\vec{\theta}_g:  \sum_{s=1}^S\mathbb{I}(\theta_{gs} \ne 0) < r\}$ and 
$\Omega_{\bar{r}}^1=\{\vec{\theta}_g:  \sum_{s=1}^S\mathbb{I}(\theta_{gs} \ne 0)\ge r\}$,
to correspond to $\HS_{\bar{r}}$.

Finally, for a declared DE gene,
we are given the posterior probability of whether gene $g$ in study $s$ is 
a non-DE gene ($\Pr(Y_{gs}=0|Z)$), 
an up-regulated gene ($\Pr(Y_{gs}=1|Z)$),
or a down-regulated gene ($\Pr(Y_{gs}=-1|Z)$). 
We propose a gene- and study-specific confidence score $V_{gs}=\Pr(Y_{gs}=1|Z)-\Pr(Y_{gs}=-1|Z)$, 
which ranges between $-1$ and $1$.
We are confident that gene $g$ is up-regulated in study $s$ if $V_{gs}$ is close to $1$, 
and vice versa when $V_{gs}$  is close to $-1$.
See figure~\ref{fig:mouse2_modulesB} for an example.

\subsection{Biomarker clustering for meta-patterns of homogenous and heterogenous differential signals}
\label{sec:tightclust}

Several recently developed meta-analysis methods \citep{li2011adaptively, bhattacharjee2012subset, li2014meta}
provide modeling of homogeneous and heterogeneous differential signals.
The 0 or 1 differential expression indicators (e.g., $\vec{w}_g$ in AW-Fisher) allow further biological investigation on consensus biomarkers as well as study-specific biomarkers.
However, when $S$ becomes large, the number of biomarker categories grows exponentially to $2^S-1$,
and the biomarker categories become intractable.
In BayesMP, the posterior probability of the differential expression indicator (i.e., $\Pr(Y_{gs}|Z)$)
 provides probabilistic soft conclusions.
After we obtain a list of biomarkers under certain global FDR control (e.g., 5\% or 1\%),
we apply the tight clustering algorithm \citep{tseng2005tight} to generate data-driven biomarker modules.
Tight clustering is a resampling-based algorithm built upon $K$-means or $K$-medoids.
This method aggregates information from repeated clustering of subsampled data to directly identify tight clusters 
(i.e., sets of biomarkers with small dissimilarity)
and does not force every biomarker into a cluster. 
It can be applied to any dissimilarity matrix 
if $K$-medoids is used. 
Pathway enrichment analysis is then applied to functionally annotate each biomarker module.
The resulting biomarker modules of different meta-patterns will greatly facilitate interpretation and hypothesis generation for further biological investigation.
In the first two real data applications,
for example,
heterogeneous meta-patterns of biomarkers are expected from the nature of multi-tissue or multi-brain-region design across studies.
Biomarkers up-regulated in one brain region but non-DE (or even down-regulated) in another brain region is of great interest.
It should be noted that in the third breast cancer data application, 
meta-patterns still help characterize the heterogeneity of different cohorts (e.g., differences of study population and probe design), %, strength of differential expression signals), 
even though this heterogeneity is hopefully minimal since these studies focus on the same disease with homogeneous tissue type.

To apply tight clustering,
we need to define a dissimilarity measure for any pair of genes.
Denote by $\vec{U}_{gs}$ the posterior probability vector for $Y_{gs}$:
$\vec{U}_{gs}=(\Pr(Y_{gs}=1|Z), \Pr(Y_{gs}=-1|Z), \Pr(Y_{gs}=0|Z))$,
which can be estimated by MCMC samples. 
For two genes $i$ and $j$,
we first calculate the dissimilarity of $\vec{U}_{is}$ and $\vec{U}_{js}$ in study $s$ and then average over study index $s$.
The dissimilarity measure between $\vec{U}_{is}$ and $\vec{U}_{js}$ we considered includes cosine dissimilarity, $l_2$ dissimilarity, $l2_{2D}$ dissimilarity, symmetric KL dissimilarity, and Hellinger dissimilarity.
Definitions and details of these dissimilarity measurements are in supplementary~section~\Rmnum{2} \citep{supp}.
By using the simulation setting in section~\ref{s:simu2}, we found cosine dissimilarity outperforms others [see details in supplementary figure~S3 \cite{supp}], 
and hence adopted it in our paper and would recommend it for other applications.

\section{Simulation results}
\label{sec:simu}

	\subsection{DE gene detection and FDR control}
        	 \label{s:simu1}

	To evaluate the performance of the proposed method and compare to other methods,
	we performed the simulations below:

	\begin{enumerate}
	\item Let $S$ be the number of studies, 
	$G=10,000$ be the total number of genes, 
	and $N=20$ be the number of cases and controls ($2N$ samples in total).
	\item We firstly focus on simulating gene correlation structure and assume no effect size for all genes in all studies.
	We sample expression levels with correlated genes following the procedure in \citet{song2014hypothesis}.
	\begin{enumerate}
	\item Sample 200 gene clusters with 20 genes in each cluster,
	and the remaining 6,000 genes are uncorrelated.  
	Denote by $C_g\in \{0, 1, \ldots, 200\}$ the cluster membership indicator for gene $g$, 
	for example, $C_g=1$ indicates gene $g$ is in cluster $1$, 
	whereas $C_g=0$ indicates gene $g$ is not in any gene cluster.
	\item For cluster $c$ and study $s$, 
	sample $A'_{cs}\sim\mbox{W}^{-1}(\Psi, 60)$, 
	where $1\le c\le 200$, $\Psi=0.5I_{20\times 20}+0.5J_{20\times 20}$, 
	$\mbox{W}^{-1}$ denotes the inverse Wishart distribution, 
	$I$ is the identity matrix, and $J$ is the matrix with all elements equal to $1$.  
	Then $A_{cs}$ is calculated by standardizing $A'_{cs}$ such that the diagonal elements are all $1$s.  
	The covariance matrix for gene cluster $c$ in study $s$ is calculated as $\Sigma_{cs}=\sigma^2A_{cs}$,
	where $\sigma$ is a tuning parameter we vary in the evaluation.
	\item 
	\label{item:Xprime}
	Denote by $g_{c1},\ldots,g_{c20}$ the indices of the 20 genes in cluster $c$ (i.e. $C_{g_{cj}}=c$, 
	where $1\le c\le C (C = 200)$,  and $1\le j\le20$).  
	Sample expression levels of genes in cluster $c$ for sample $n$ in study $s$ as $(X'_{g_{c1}sn},\ldots,X'_{g_{c20}sn})\sim \mbox{MVN}(0,\Sigma_{cs})$, 
	where $1\le n\le 2N$ and $1\le s\le S$.  
	For any uncorrelated gene $g$ with $C_g=0$, 
	sample the expression level for sample $n$ in study $s$ as $X'_{gsn}\sim \mbox{N}(0, \sigma^2)$, 
	where $1\le n\le 2N$ and $1\le s\le S$.
	\end{enumerate}
	\item Sample DE genes, effect sizes, and their differential expression directions.
	\begin{enumerate}
	\item Assume that the first $G_1$ genes are DE in at least one of the combined studies, 
	where $G_1=30\% \times G$.  
	For each $1\le g\le G_1$, 
	sample $v_g$ from discrete uniform distribution $v_g\sim \mbox{UNIF}(1,\ldots,S)$,
	and then randomly sample subset  $\mathbf{v}_g\subseteq\{1,\ldots,S\}$ such that $|\mathbf{v}_g|=v_g$.  
	Here $\mathbf{v}_g$ is the set of studies in which gene $g$ is DE. 
	\item For any DE gene $g$ ($1\le g\le G_1$), sample  gene-level effect size $\theta_g\sim \mbox{N}_{0.5+}(1,1)$, where $\mbox{N}_{a+}$ denotes the truncated Gaussian distribution within interval $(a, \infty)$.  For any $s\in \mathbf{v}_g$, also sample study-specific effect size $\theta_{gs}\sim \mbox{N}_{0+}(\theta_g, 0.2^2)$.
	\label{step:effectSize}
	\item 
	Sample $d_g\sim \mbox{Ber}(0.5)$, where $1\le g\le G_1$ and $s\in \mathbf{v}_g$.  Here $d_g$ controls effect size direction for gene $g$.
	\end{enumerate}
	\item Add the effect sizes to the gene expression levels sampled in step~\ref{item:Xprime}.  For controls ($1\le n \le N$), set the expression levels as $X_{gsn}=X'_{gsn}$. For cases ($N+1\le n\le 2N$), if $1\le g\le G_1$ and $s\in\mathbf{v}_g$, set the expression levels as $X_{gsn}=X'_{gsn}+(-1)^{d_g}\theta_{gs}$; otherwise, set $X_{gsn}=X'_{gsn}$.
	\end{enumerate}
	We performed simulation with $S=3,5,10$ and $\sigma=1, 2, 3$ to account for different numbers of combined studies and various signal/noise ratio.  
	We applied \textit{limma} to compare the gene expression levels between the control group and the case group.
	We transformed the two-sided \textit{p}-values from \textit{limma} to one-sided \textit{p}-values by taking account of the directions of estimated effect sizes. 
	Then one-sided \textit{p}-values are transformed to Z statistics.	
	BayesMP took 53 minutes on a regular PC with  1.4 GHz CPU 
	(i.e., for one simulation with $S=3$ and $\sigma=1$) to obtain
	10,000 posterior samples using MCMC.
	Supplementary figure~S1 \citep{supp} shows the posterior samples of $\pi_g$ in two example genes --- a DE gene and a non-DE gene as well as $\gamma$.
	Because the posterior samples converge to a stationary distribution very quickly for our method [see examples in supplementary figure S1 \citep{supp}], excluding 500 posterior samples for burn-in is enough for the analyses of this paper.
	We repeated the simulation $100$ times and averaged the results.	
		We compared the performance of our method and existing methods designed for decision space $\D_{\bar{A}}$ (maxP), 
		$\D_B$ (Fisher's method and AW), 
		and $\D_{\bar{r}}$ (rOP) with $r=\lfloor S/2\rfloor + 1$ using false discovery rate (FDR), false negative rate (FNR), and the area under the curve (AUC) of the receiver operating characteristic (ROC) curve.  
		Note that in our comparison, we used $\D_{\bar{A}}$ and $\D_{\bar{r}}$ which are equivalent to the complementary hypothesis testings $\HS_{\bar{A}}$ and $\HS_{\bar{r}}$,
		 and the true number of studies in which a gene is DE can be calculated because the truth is known in simulation.  
	Table~\ref{tab:fdr}  compares the FDR, FNR, and AUC of different methods
	at nominal FDR level 5\%, which is widely accepted in genomic research.
	For decision space $\D_B$, all the three methods controlled FDR around its nominal level, 
	which is anticipated because $\HS_B$ is complementary and equivalent to $\D_B$.  
	Fisher and AW were slightly overconservative in terms of FDR control --- 
	Fisher's method and AW controlled FDR at around 3.5\%, 
	whereas our BayesMP controlled FDR at around the nominal 5\%.  
	This phenomenon has also been observed in \citet{song2014hypothesis} when the genes are correlated.  
	Because BayesMP is less conservative than the other two methods, 
	we were able to detect slightly more genes under $\D_B$.	
	In addition, BayesMP achieved similar (or slightly better) FNR and AUC with Fisher and AW under  $\D_B$.
	Fisher' s method is known to be almost optimal --- that is, 
	Fisher' s method achieves asymptotic Bahadur optimality (ABO) \citep{littell1971asymptotic} when effect sizes are consistent and equal for all studies.
	This indicated BayesMP is also almost optimal for  $\D_B$ under the simulation scenario.
	For decision space $\D_{\bar{A}}$ and $\D_{\bar{r}}$, 
	we observed that maxP and rOP lost control of FDR.  
	As discussed in section \ref{s:intro}, 
	this is caused by the nature that $\HS_A$ and $\HS_r$ have noncomplementary null and alternative spaces. 
	 To the contrary, BayesMP still controlled FDR close to its nominal level for $\D_{\bar{A}}$ and $\D_{\bar{r}}$.  
	 Note that because maxP and rOP were not able to control FDR at its nominal level, 
	 the number of genes detected by these methods was not directly comparable to our methods.  
	 However, FNR of BayesMP was only slightly larger than rOP for $\D_{\bar{r}}$, 
	 regardless of the conservative FDR control.
	 When $S$ was large ($S=10$) in simulation,
	 the FDR control of BayesMP under $\D_{\bar{A}}$ deviated from its nominal level (around 10\% instead of 5\%).
	 The reason for the anticonservative control was that the data simulation setting was different from model generative process,
	 thus small errors  accumulated when $S$ got large.
	 However, BayesMP still performed much better than maxP and roP (FDR = 0.58 for maxP in $\D_{\bar{A}}$ 
	 and FDR = 0.23 for rOP in $\D_{\bar{r}}$ setting).
	In addition, BayesMP achieved much larger AUC than maxP and rOP under  $\D_{\bar{A}}$ and $\D_{\bar{r}}$,
	which indicated better predictive power of BayesMP.

\begin{table}
\centering
\caption{Comparison of different methods by FDR, FNR, and AUC of ROC curve for decision spaces 
$\D_{\bar{A}}$, $\D_B$, and $\D_{\bar{r}}$.  
The nominal FDR is 5\% for all compared methods.  
The mean results and SD (in parentheses) were calculated based on $100$ simulations.} 
\label{tab:fdr}
\scriptsize
  \begin{tabular}{c|cc|cc|ccc|cc}
   \hline
   \hline
 &  & & \multicolumn{2}{c|}{$\D_{\bar{A}}$} & \multicolumn{3}{c|}{$\D_B$} & \multicolumn{2}{c}{$\D_{\bar{r}}$ ($r=\lfloor S/2 \rfloor + 1$)}\\
 &   $S$ & $\sigma$ & BayesMP & maxP & BayesMP & Fisher & AW & BayesMP & rOP\\
   \hline
   \multirow{18}{*}{FDR} 
&   \multirow{6}{*}{$3$}  & \multirow{2}{*}{$1$}  & 0.058 & 0.207 & 0.042 & 0.035 & 0.035 & 0.034 & 0.086\\
&&& (0.008) & (0.014) & (0.004) & (0.005) & (0.004) & (0.004) & (0.007)\\
			&				   &\multirow{2}{*}{$2$}  & 0.058 & 0.198 & 0.047 & 0.035 & 0.036 & 0.037 & 0.080\\
&&& (0.010) & (0.017) & (0.006) & (0.006) & (0.006) & (0.005) & (0.009)\\
			    
&				   &\multirow{2}{*}{$3$}  & 0.043 & 0.184 & 0.050 & 0.035 & 0.036 & 0.036 & 0.073\\
&&& (0.016) & (0.025) & (0.009) & (0.008) & (0.009) & (0.009) & (0.014)\\

\cline{2-10}			    
 &   \multirow{6}{*}{$5$}  & \multirow{2}{*}{$1$}  & 0.075 & 0.361 & 0.043 & 0.034 & 0.034 & 0.037 & 0.130\\
&&& (0.011) & (0.017) & (0.004) & (0.004) & (0.004) & (0.005) & (0.008)\\
			&				   &\multirow{2}{*}{$2$}  & 0.079 & 0.349 & 0.046 & 0.034 & 0.034 & 0.042 & 0.115\\
&&& (0.019) & (0.022) & (0.006) & (0.005) & (0.005) & (0.006) & (0.010)\\
			    
&				   &\multirow{2}{*}{$3$}  & 0.062 & 0.330 & 0.050 & 0.034 & 0.034 & 0.042 & 0.099\\
&&& (0.033) & (0.028) & (0.007) & (0.006) & (0.007) & (0.008) & (0.013)\\
			    
\cline{2-10}			    
 &   \multirow{6}{*}{$10$}  & \multirow{2}{*}{$1$}  & 0.105 & 0.580 & 0.050 & 0.035 & 0.035 & 0.047 & 0.231\\
&&& (0.017) & (0.021) & (0.003) & (0.004) & (0.004) & (0.005) & (0.010)\\
			&				   &\multirow{2}{*}{$2$}  & 0.122 & 0.569 & 0.050 & 0.035 & 0.035 & 0.059 & 0.200\\
&&& (0.029) & (0.027) & (0.005) & (0.005) & (0.005) & (0.007) & (0.012)\\
			    
&				   &\multirow{2}{*}{$3$}  & 0.108 & 0.554 & 0.054 & 0.035 & 0.035 & 0.062 & 0.167\\
&&& (0.064) & (0.044) & (0.007) & (0.006) & (0.006) & (0.010) & (0.015)\\
			    
   \hline
   \hline
   \multirow{18}{*}{FNR} 
&   \multirow{6}{*}{$3$}  & \multirow{2}{*}{$1$}  & 0.024 & 0.017 & 0.054 & 0.058 & 0.056 & 0.039 & 0.032\\
&&& (0.002) & (0.001) & (0.003) & (0.003) & (0.003) & (0.002) & (0.002)\\
			&				   &\multirow{2}{*}{$2$}  & 0.064 & 0.055 & 0.170 & 0.181 & 0.183 & 0.114 & 0.112\\
&&& (0.002) & (0.002) & (0.003) & (0.003) & (0.003) & (0.003) & (0.003)\\
			    
&				   &\multirow{2}{*}{$3$}  & 0.089 & 0.082 & 0.240 & 0.253 & 0.257 & 0.161 & 0.166\\
&&& (0.003) & (0.002) & (0.002) & (0.002) & (0.002) & (0.003) & (0.003)\\

\cline{2-10}			    
 &   \multirow{6}{*}{$5$}  & \multirow{2}{*}{$1$}  & 0.016 & 0.009 & 0.043 & 0.048 & 0.045 & 0.032 & 0.021\\
&&& (0.001) & (0.001) & (0.003) & (0.002) & (0.003) & (0.002) & (0.001)\\
			&				   &\multirow{2}{*}{$2$}  & 0.040 & 0.031 & 0.150 & 0.162 & 0.164 & 0.092 & 0.086\\
&&& (0.002) & (0.002) & (0.003) & (0.003) & (0.003) & (0.002) & (0.002)\\
			    
&				   &\multirow{2}{*}{$3$}  & 0.054 & 0.047 & 0.219 & 0.236 & 0.241 & 0.132 & 0.136\\
&&& (0.002) & (0.002) & (0.003) & (0.003) & (0.003) & (0.002) & (0.003)\\
			    
\cline{2-10}			    
 &   \multirow{6}{*}{$10$}  & \multirow{2}{*}{$1$}  & 0.009 & 0.004 & 0.030 & 0.035 & 0.031 & 0.027 & 0.010\\
&&& (0.001) & (0.001) & (0.002) & (0.002) & (0.002) & (0.001) & (0.001)\\
			&				   &\multirow{2}{*}{$2$}  & 0.022 & 0.015 & 0.119 & 0.132 & 0.132 & 0.070 & 0.054\\
&&& (0.001) & (0.001) & (0.003) & (0.003) & (0.003) & (0.002) & (0.002)\\
			    
&				   &\multirow{2}{*}{$3$}  & 0.028 & 0.023 & 0.184 & 0.206 & 0.211 & 0.097 & 0.095\\
&&& (0.002) & (0.002) & (0.003) & (0.002) & (0.003) & (0.003) & (0.003)\\
			    
   \hline
   \hline
   \multirow{18}{*}{AUC} 
&   \multirow{6}{*}{$3$}  & \multirow{2}{*}{$1$}  & 0.976 & 0.926 & 0.973 & 0.973 & 0.973 & 0.980 & 0.972\\
&&& (0.003) & (0.003) & (0.002) & (0.002) & (0.002) & (0.002) & (0.003)\\
			&				   &\multirow{2}{*}{$2$}  & 0.906 & 0.876 & 0.880 & 0.878 & 0.876 & 0.902 & 0.873\\
&&& (0.006) & (0.006) & (0.005) & (0.005) & (0.005) & (0.005) & (0.006)\\
			    
&				   &\multirow{2}{*}{$3$}  & 0.833 & 0.806 & 0.788 & 0.784 & 0.780 & 0.820 & 0.776\\
&&& (0.008) & (0.008) & (0.006) & (0.006) & (0.006) & (0.007) & (0.008)\\

\cline{2-10}			    
 &   \multirow{6}{*}{$5$}  & \multirow{2}{*}{$1$}  & 0.974 & 0.920 & 0.978 & 0.978 & 0.979 & 0.985 & 0.979\\
&&& (0.004) & (0.003) & (0.002) & (0.002) & (0.002) & (0.002) & (0.002)\\
			&				   &\multirow{2}{*}{$2$}  & 0.918 & 0.891 & 0.896 & 0.893 & 0.892 & 0.928 & 0.893\\
&&& (0.007) & (0.006) & (0.004) & (0.004) & (0.004) & (0.004) & (0.004)\\
			    
&				   &\multirow{2}{*}{$3$}  & 0.866 & 0.833 & 0.812 & 0.805 & 0.800 & 0.859 & 0.801\\
&&& (0.009) & (0.009) & (0.005) & (0.005) & (0.006) & (0.005) & (0.006)\\
			    
\cline{2-10}			    
 &   \multirow{6}{*}{$10$}  & \multirow{2}{*}{$1$}  & 0.964 & 0.910 & 0.985 & 0.983 & 0.985 & 0.985 & 0.986\\
&&& (0.007) & (0.003) & (0.002) & (0.002) & (0.002) & (0.002) & (0.002)\\
			&				   &\multirow{2}{*}{$2$}  & 0.907 & 0.905 & 0.920 & 0.917 & 0.917 & 0.948 & 0.920\\
&&& (0.010) & (0.006) & (0.004) & (0.004) & (0.004) & (0.004) & (0.004)\\
			    
&				   &\multirow{2}{*}{$3$}  & 0.883 & 0.865 & 0.849 & 0.840 & 0.835 & 0.907 & 0.838\\
&&& (0.011) & (0.009) & (0.005) & (0.005) & (0.005) & (0.005) & (0.006)\\
			    
   \hline   
  \end{tabular}
\end{table}

	\subsection{Simulation to evaluate  meta-pattern gene module detection}
        	 \label{s:simu2}

	To evaluate the performance of gene module detection,
	we adopted a simulation procedure similar to section \ref{s:simu1}.  
	We simulated $S=4$ studies in total.
	Among the $G = 10,000$ genes, we set 
	$4\%$ of them as homogeneously concordant DE genes, 
	 with the same direction in all studies (all positive or all negative).
	We denote ``homo$+$'' as the homogeneously concordant DE genes with all positive effect sizes and 
	``homo$-$'' as the homogeneously concordant DE genes with all negative effect sizes.
	We also set another $4\%$ of all genes as study-specific DE genes --- differentially expressed only in one study.
	Among them, 
	$1/4$ are DE genes only in the first study with positive effect sizes (denoted as ``ssp$1+$''),
	$1/4$ are DE genes only in the first study with negative effect sizes (denoted as ``ssp$1-$''),
	$1/4$ are DE genes only in the second study with positive effect sizes (denoted as ``ssp$2+$''),
	and the remaining $1/4$ are DE genes only in the second study with negative effect sizes (denoted as ``ssp$1-$'').
	The rest of the genes are not DE (denoted as ``non-DE'').
	The biological variance $\sigma$ is set to $1$ in this simulation.
	
        We first applied the proposed method to this synthetic dataset.
        We controlled FDR at $5\%$ under $\D_B$ and obtained 691 genes.
        These genes were used as input for our gene module detection using the tight clustering algorithm.
	We identified six gene modules in these 691 genes.  
	The detected gene modules are tabulated against the true gene modules simulated in table~\ref{tab:simu2_tight}
	(module 0 contains scattered genes not assigned to any of the six modules).
	The detected gene modules clearly correspond to the true modules, 
	 and most of the non-DE genes were left to the noises.  
	 The heatmaps, confidence scores and DE patterns of these six modules are shown in supplementary figure~S2 \citep{supp}.
An alternative approach is to apply tight clustering directly on the \textit{Z}-statistics.  
By comparing the results, we found that the modules detected by this naive approach are neither pure nor distinguishable under our simulation settings [see details in supplementary table~S1 \citep{supp}].
				
\begin{table}
	\centering
\caption{Contingency table of simulation underlying truth and tight clustering result with 6 target modules. 
0 represents the scattered gene group.
1 $\sim$ 6 represent 6 detected modules.
Bolded numbers are genes with correct assignment.
}
\label{tab:simu2_tight}
\begin{tabular}{c|ccccccc}
\hline
\hline
Module & homo$-$ & homo$+$ & ssp$1-$ & ssp$1+$ & ssp$2-$ & ssp$2+$ & non-DE\\
\hline
1 & \textbf{177} & 0 & 0 & 0 & 0 & 0 & 0\\
2 & 0 & \textbf{164} & 0 & 0 & 0 & 0 & 0\\
3 & 0 & 2 & 0 & 0 & 0 & \textbf{84} & 4\\
4 & 0 & 4 & 0 & \textbf{72} & 0 & 0 & 6\\
5 & 0 & 0 & \textbf{66} & 0 & 0 & 0 & 2\\
6 & 1 & 0 & 0 & 0 & \textbf{62} & 0 & 0\\
0 & 6 & 4 & 0 & 0 & 7 & 0 & \textbf{19}\\

\hline
\end{tabular}

\end{table}%

	\subsection{Additional simulations on sample size effects}
	\label{sec:simu3}
	To assess impact of unbalanced sample size, we simulated the following special scenarios with 
	\begin{enumerate}
		\item different numbers of samples in different studies,
		\label{step:2}
		\item different numbers of cases and controls in each study,
		\label{step:3}
		\item different ratios of case and control samples in each study.
		\label{step:4}
	\end{enumerate}
	Below, we followed the simulation setting in section~\ref{s:simu1} unless otherwise mentioned.
	In scenario~\ref{step:2}, 
	we allowed different studies to have different numbers of samples.
	Under this scenario, we simulated 
	case (a), with the numbers of samples (case/control) being 20/20, 30/30, 40/40 for three studies respectively,
	and case (b), with the numbers of samples (case/control) being 20/20, 50/50, 100/100  respectively.
	In scenario~\ref{step:3}, 
	we allowed the numbers of cases and controls to be different within each study.
	Under this scenario, we simulated 
	case (c), with the numbers of samples (case/control) being 60/20, 60/20, 60/20 for three studies respectively.
	In scenario~\ref{step:4}, 
	we allowed the ratios of case and control samples to be different within each study.
	Under this scenario, we simulated 
	case (d), with the numbers of samples (case/control) being 20/60, 40/40, 60/20 for three studies respectively.
	
	The results of simulation cases (a)-(d) are shown in supplementary table~S2 \citep{supp}.
	We observed that under scenario~\ref{step:2}, scenario~\ref{step:3}, and scenario~\ref{step:4}, 	
	BayesMP controlled FDR to its nominal level for $\D_{\bar{A}}$, $\D_B$, and $\D_{\bar{r}}$.
	These results indicate that our Bayesian model is robust against impact of heterogeneity sample size in a wide spectrum of scenarios.

	\subsection{Additional simulations on robustness of the algorithm}
	\label{sec:robustness}
	In our Bayesian hierarchical model, 
	we assume that the null component $f_0$ comes from $\mbox{N}(0,1)$ --- the theoretical null for all studies.
	However, this assumption can be violated when genes from null components are correlated \citep{efron2001empirical}.
	Therefore, we designed simulations to access the performance of our model when theoretical null assumption is not valid. 
	To be specific,  in our simulation setting step~\ref{item:Xprime}, 
	we varied the number of correlated clusters $C = 200, 300, 400, 500$,
	representing increasing probability of correlated null component. 
	We evaluated the performance of the Bayesian hierarchical model, 
	and the result is shown in supplementary table~S4 \citep{supp}.
	We observe that the BayesMP still performs very well even though the null components are correlated.
	Therefore BayesMP is robust against the theoretical null assumption for the null component.

\section{Real data applications}
\label{sec:realData}
To further evaluate our method and demonstrate its usage,
we applied BayesMP on three real meta-analysis examples:
one on the gene expression of  multi-tissue microarray studies using metabolism-related knockout mice,
one on multi-brain-region RNA-seq studies using HIV transgenic rats,
and another on transcriptomic breast cancer studies across multiple platforms.
The sample size description is shown in supplementary table~S5 \citep{supp}.

	\subsection{Mouse metabolism data}
        \label{s:mouseReal}
Very long-chain acyl-CoA dehydrogenase (VLCAD) deficiency was found to be associated with energy metabolism disorder in children \citep{li2011adaptively}. 
Two genotypes of the mouse model --- wild type (VLCAD $+$/$+$) and VLCAD-deficient (VLCAD $-$/$-$) --- 
were studied for three types of tissues (brown fat, liver, and heart) with three to four mice in each genotype group.  
The total number of genes from these three transcriptomic microarray studies is 14,495.
Supplementary table~S5(a) \citep{supp} shows details of the study design.
Two-sided \textit{p}-values were calculated using \textit{limma} by comparing wild-type versus VLCAD-deficient mice in each tissue,
and one-sided \textit{p}-values were obtained by considering the effect size direction.
BayesMP took 62 minutes to obtain 10,000 posterior samples using MCMC, 
and the first 500 posterior samples were excluded as burn-in iterations.
By controlling FDR at $5\%$,
we detected 168 probes under $\D_{\bar{A}}$; 
among them, 156 have concordant effect size directions in all the three tissues.
The heatmap for the genes detected under $\D_{\bar{A}}$ is shown in supplementary figure~S6 \citep{supp}.

        Similarly, under $\D_B$ we {\color{black} obtained 3,496 DE genes at an FDR level of 5\% and 1,243 DE genes at an FDR level of 1\%.  
        Due to the unusually strong genome-wide biological signal,  
        we decided to use the FDR cutoff of 1\% for the downstream analysis to increase statistical power (of the Fisher's exact test) in pathway enrichment analysis. }
        Then we applied the tight clustering algorithm using the cosine distance as described in section \ref{sec:tightclust} to detect modules based on 1,243 DE genes at the stringent FDR level of 1\%.
	The results are shown in figure~\ref{fig:mouse1_modules}.
	Using the tight clustering, we were able to detect 6 gene modules with unique patterns.
	The first two biomarker modules are consensus genes that are up-regulated or down-regulated in all tissues.
	The next four modules are biomarkers with study-specific differential patterns.
	For example,
	DE genes in module \Rmnum{3} are up-regulated in the heart but not in the brown fat or the liver.
	To examine the biological functions of these modules, 
	we performed pathway enrichment analysis for genes in each module using Fisher's exact test.
	The pathway database was downloaded from the Molecular Signatures Database (MSigDB) v5.0 (\url{http://bioinf.wehi.edu.au/software/MSigDB/}),
	where a mouse-version pathway database was created by combining pathways from KEGG, BIOCARTA, REACTOME, and GO databases and mapping all the human genes to their orthologs in mouse
	using Jackson Laboratory Human and Mouse Orthology Report 
	(\url{http://www.informatics.jax.org/orthology.shtml}).
	{\color{black}
	The resulting \textit{p}-values were converted to $q$-values by Benjamini-Hochberg correction \citep{benjamini1995controlling} to adjust for multiple comparison,
	where $q$-value measures the false discovery rate (FDR) one would incur by accepting the given test.
	}
	At an FDR cutoff of $5\%$, 
	we summarized the pathway detection result 
	(see supplementary Excel file 1 for detailed pathway information).
	Among the six gene modules with distinct DE patterns, 
	module \Rmnum{1} is enriched in enzyme pathways (e.g., KEGG lysosome; $q=1.6\times 10^{-3}$),
	module \Rmnum{2} is enriched in pathways for lyase activity
	(e.g., GO lyase activity; $q=0.26$),
	module \Rmnum{3} is enriched in defense response pathways
	(e.g., GO defense response pathway; $q=2.2\times 10^{-6}$),
	module \Rmnum{4} is enriched in phosphatase regulator  pathways
	(e,g, GO phosphatase regulator activity;  $q=0.12$),
	module \Rmnum{6} is enriched in platelet related pathways
	(e.g., GO formation of platelet; $q=2.2\times 10^{-2}$).
	For module \Rmnum{5}, we didn't detect any enriched pathways.
	Remarkably, all of these pathways are known to be related to different aspects of metabolism, 
	which indicates that our method is able to detect homogeneous and heterogeneous gene modules that are biologically meaningful.
	The biomarker clustering result enhances meta-analysis interpretation and 
	motivates hypothesis for further biological investigation.
	For example, 
	it is intriguing to understand why VLCAD-mutation impacts DE genes only up-regulated in the heart but not in the brown fat or the liver, 
	and why these genes are associated with the defense response pathway.

\begin{figure}[htbp]
	\centering
	\subfigure[Heatmap]{
		\label{fig:mouse1_modulesA}
		\includegraphics[height=1.1\columnwidth]{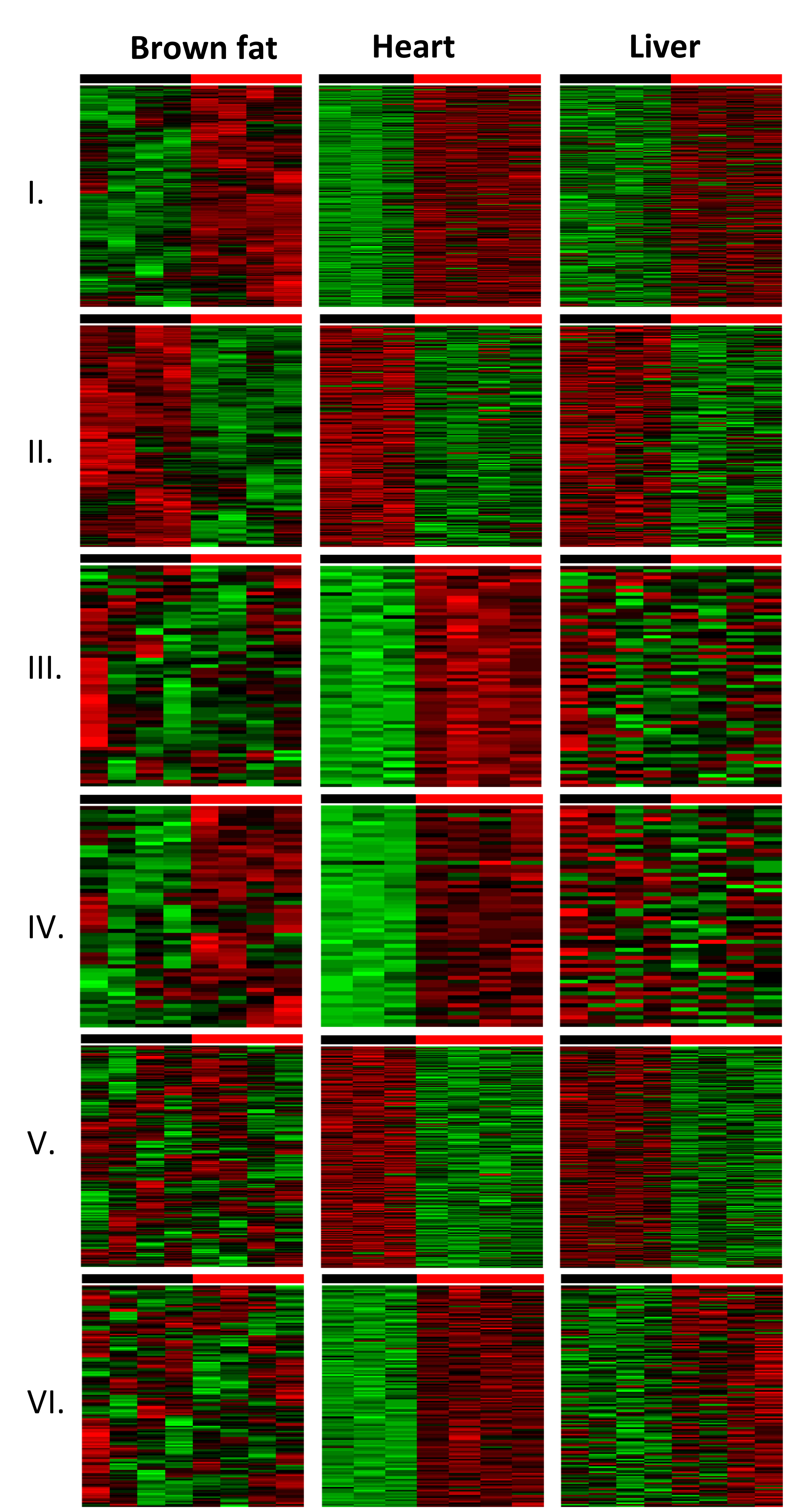}
	}
	\subfigure[CS]{
		\label{fig:mouse1_modulesB}
		\includegraphics[height=1.1\columnwidth]{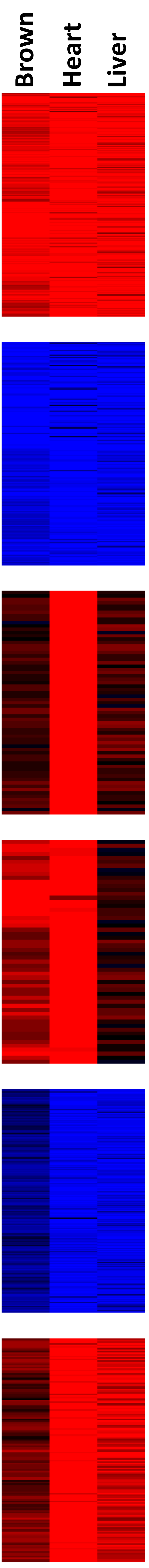}
	}
	\subfigure[bar plot]{
		\label{fig:mouse1_modulesC}
		\includegraphics[height=1.1\columnwidth]{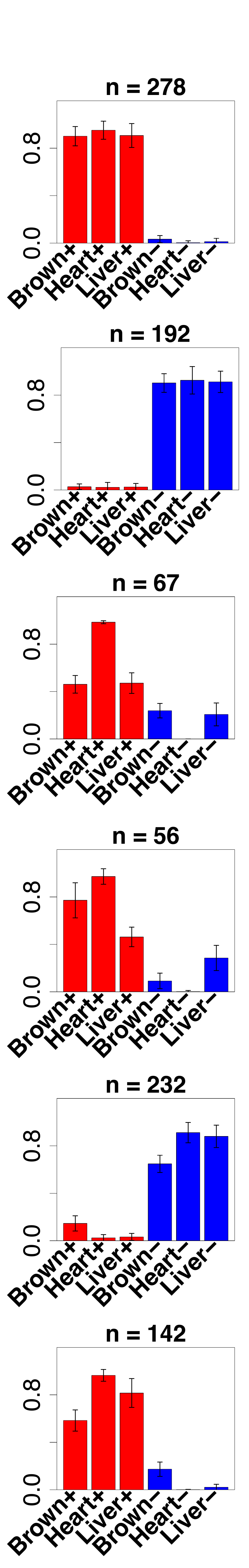}
	}
		\caption{
	Six meta-pattern modules of biomarkers from the mouse metabolism example.
	Each row shows a set of detected biomarkers showing similar meta-pattern of differential signals.
	\ref{fig:mouse1_modulesA} Heatmaps of detected genes (on the rows) and samples (on the columns) for each tissue (brown fat, heart, or liver), 
	where each tissue represents a study (i.e., brown fat for $s=1$, heart for $s=2$, liver for $s=3$).	
	The black color bar on top represents wild type (control), 
	and the red color bar on top represents VLCAD-deficient mice (case).
	\ref{fig:mouse1_modulesB} Heatmaps of confidence scores (CS) (genes on the rows and three studies on the columns). 
	Confidence score is described in section~\ref{sec:decision}, which ranges from $-1$ (blue color for down-regulation) to 1 (red color for up-regulation).
	\ref{fig:mouse1_modulesC} Bar plots of mean posterior probability for $Y_{gs} =1$ (red color for up-regulation) and $Y_{gs} = -1$ (blue color for down-regulation) for each module in each tissue.
	Error bar represents standard deviation across all genes in the module.
	The number of genes is shown on top of each bar plot.
	In figure~\ref{fig:mouse1_modulesB} and  \ref{fig:mouse1_modulesC}, we use ``Brown" to denote ``Brown fat".
	}
  	\label{fig:mouse1_modules}
\end{figure}

	\subsection{HIV transgenic rat RNA-seq data}
        \label{s:mouseHIV}

\citet{li2013transcriptome} conducted studies to determine gene expression differences 
between F344 and HIV transgenic rats using RNA-seq (GSE47474 in the Gene Expression Omnibus database 
[\url{http://www.ncbi.nlm.nih.gov/geo/query/acc.cgi?acc=GSE47474}]).  
The HIV transgenic rat model was designed to study learning, memory, vulnerability to drug addiction, 
and other psychiatric disorders vulnerable to HIV-positive patients.
They sequenced RNA transcripts with 12 F334 rats and 12 HIV transgenic rats 
in prefrontal cortex (PFC), hippocampus (HIP), and striatum (STR) brain regions
[see detail in supplementary table~S5(b) \citep{supp}].
We applied the same alignment procedure using TopHat \citep{trapnell2009tophat} adopted by \citet{li2013transcriptome} and
obtained the RNA-seq count data for 16,821 genes by BEDTools \citep{quinlan2010bedtools}.
We filtered out genes with less than 100 total counts within any brain region and ended up with 11,824 genes.
We removed potential outliers by checking the sample correlation heatmaps (supplementary figure~S7) \citep{supp}.
We employed R package \textit{edgeR}  to perform DE gene detection and obtained two-sided \textit{p}-values.
The one-sided \textit{p}-values were obtained by considering the effect size directions and further converted to Z statistics.
It took 41 minutes to obtain 10,000 posterior samples via MCMC, 
and the  first 500 posterior samples were excluded as burn-in iterations.
Since it is well known that the postmortem brain expression profiles generally contain weak signals,
we controlled  FDR at $20\%$ in the analysis.  
Under $\D_{\bar{A}}$, we detected 69 genes, of which all 69 had concordant DE directions.
The heatmaps of the expression levels (log of normalized counts) of these genes in the three brain regions are shown in supplementary figure~S8 \citep{supp}.  
Under $\D_B$, we detected 669 genes.
We further applied the tight clustering algorithm, 
and obtained 3 gene modules.
Their gene expression heatmaps, 
DE confidence scores and bar plots of posterior probability of differential expression are shown in
figure~\ref{fig:mouse2_modules}.
To examine the biological functions of these modules, 
we also performed pathway enrichment analysis using the same procedure as in section~\ref{s:mouseReal}
(see supplementary Excel file 2 for detailed information).  
{\color{black}
As the postmortem brain expression profiles generally contain much weaker signals and the gene size of each module is relatively small,
we presented \textit{p}-values (unadjusted for multiple comparison) instead of $q$-values for the below pathway enrichment analysis.
}
According to the results, module \Rmnum{1} is down-regulated in all three brain regions, 
and is enriched in pathways related to the immune system
(e.g., REACTOME inntate immunity signaling; $p = 6.26 \times 10^{-3}$),
module \Rmnum{2} is up-regulated in all three brain regions 
and is enriched in pathways related to  response to virus
(e.g., GO response to virus; $p = 1.81 \times 10^{-3}$),
module \Rmnum{3} is down-regulated in HIP, 
but up-regulated in PFC and STR, and it is enriched in pathways related to synapsis 
(e.g., GO synaptic transmission; $p=2.75 \times 10^{-3}$) and 
neuron connections
(e.g., KEGG neuroactive ligand receptor interaction; $p=2.88 \times 10^{-3}$).
Since it is well-known that  HIV attacks the immune system \citep{weiss1993does}, 
we anticipate genes for immune response to be down-regulated, as observed in module \Rmnum{1}.
The up-regulation of response to virus pathway  we found in module \Rmnum{2} is reasonable since the mice are infected by the virus.
Moreover, because different brain regions have different functions, 
it is not surprising to discover some neuron-related genes that may respond differently to HIV in different brain regions
(module \Rmnum{3}).

	\subsection{Breast cancer dataset}
        \label{s:breastCancer}
In this example, we combined seven breast cancer transcriptomic datasets,
which study the same biological problem using different gene expression platforms, 
including Illuminia, Affymetrix and RNA-seq.
The phenotype of interest is the breast cancer grade,
which is defined according to the cancer cells' growth patterns as well as their appearance compared to to healthy breast cells. 
Grade \Rmnum{1} cancer cells show slow and well-organized growth patterns and they look a little bit different from healthy cells, 
while grade \Rmnum{3} cancer cells grow quickly in disorganized patterns, with many dividing to make new cancer cells,
which look very different from healthy cells.
Details of these 7 datasets are described in supplementary table~S5(c) \citep{supp}.
For each study, if multiple probes match to the same gene, we select the probe with the largest IQR \citep{gentleman2006bioinformatics} to represent the gene.
After matching the same gene symbols, 3,920 genes that appeared in all 7 studies were selected for the analysis.
For each study, we obtained two-sided \textit{p}-values using \textit{limma} (for continuous data) or \textit{edgeR} (for count data) by comparing grade \Rmnum{1} versus grade \Rmnum{3}
{\color{black} with adjustment of race, age, and gender as covariates whenever they were available}.  
We calculated one-sided \textit{p}-values by considering the direction of the effect sizes.
We applied BayesMP, and it took 31 minutes to obtain 10,000 posterior samples using MCMC, 
and the first 500 posterior samples were excluded as burn-in iterations.

Since we expect the studies combined in this application to be more homogeneous than previous examples, 
genes DE in most of the studies are our major interest and worth further investigation.  Moreover, because the studies are conducted by different groups using different platforms, we want our analysis to be robust against a couple of studies with poor quality or unspecific probe design.
We visualize the number of declared DE genes at FDR 5\% for $\D_{\bar{r}}$ ($r=1, \ldots, 7$) in figure~\ref{fig:brcaGrade}.
\begin{figure}[htbp]
	\centering
	\includegraphics[width=0.7\columnwidth]{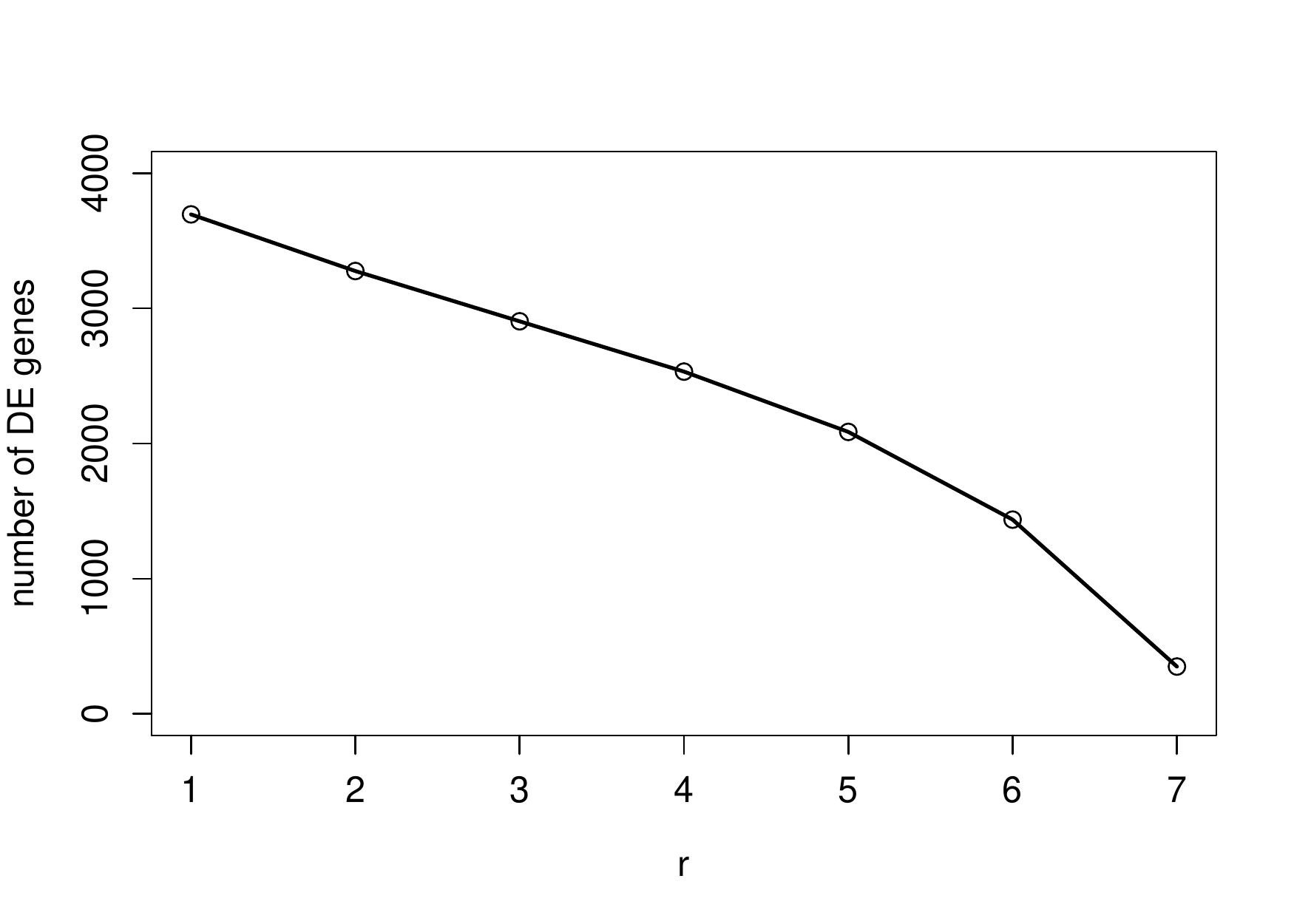}
	\caption{
	Number of declared DE genes at FDR 5\% for $\D_{\bar{r}}$ ($r=1, \ldots, 7$).
	}
	\label{fig:brcaGrade}
\end{figure}
It is noticed that there is a noticeable drop in the number of DE genes between $r=6$ and $r=7$.  
This indicates that it could be too stringent to require that DE genes agree in all seven studies. 
To make the DE gene detection replicable enough yet not too stringent, we chose to use $\D_{\bar{r}}$ ($r=6$). 
Under FDR 5\%,
we detected 1,437 significant genes.
Pathway enrichment analysis was performed using Fisher's exact test.
Top pathways associated with these genes are
REACTOME cell cycle ($q = 0.0066$) and
REACTOME DNA replication ($q = 0.0066$).
These results are biologically meaningful since it is known that cancer cells' growth patterns of different grades 
are associated with cell cycle and DNA replication.

{\color{black}
We further performed a cross-study validation to assess the performance of BayesMP.
In each iteration, we set aside one study, applied BayesMP to the remaining six studies,
and declared DE genes under $\D_{\bar{A}}$  using FDR=5\%.
To evaluate the consistency of DE gene detection results between the BayesMP of the six studies and the left-alone study,
we calculated the area under the curve (AUC) of the receiver operating characteristic (ROC) curves
by treating the DE status from BayesMP as a binary outcome and using the \textit{p}-values from the left-alone study to determine the moving sensitivity and specificity.
We repeated this procedure for each left-alone study to calculate the AUC values.
As a baseline contrast, for each left-alone single study, we also similarly calculated AUC values using the DE status (using FDR = 5\%) from each of the other six individual studies,
and calculated the average and standard error of AUCs.
As shown in supplementary table 6 \citep{supp}, 
the AUCs from BayesMP were between 0.65 and 0.85, consistently higher than the average AUCs from six individual studies (ranging from 0.59 to 0.75),
which shows quantitative validity of the meta-analysis. 
Study GSE6532 has a relatively lower cross-study validation AUC from BayesMP ($ =0.65$) than other studies, 
indicating lower compatibility with other studies.
This also justifies the usage of $r=6$ for $\D_{\bar{r}}$ in the previous paragraph. 

}

\section{Conclusion}
\label{sec:conclusion}
For meta-analysis at the genome-wide level,
the issues to efficiently integrate information across studies and genes and 
to quantify homogeneous and heterogeneous DE signals across studies have brought new statistical challenges.
The Bayesian hierarchical model provides a feasible and effective solution.
Compared to traditional hypothesis testing,
decision theory framework from Bayesian modeling provides a more flexible inference to determine DE genes from meta-analysis.
In this paper,
we proposed a Bayesian hierarchical model for general transcriptomic meta-analysis.
From posterior distribution of the latent variable (DE indicators $Y_{gs}$),
 FDR is well controlled, 
 and there is no need to select different test statistics for different {\color{black} hypothesis settings
($\HS_{\bar{A}}$, $\HS_B$, and $\HS_{\bar{r}}$).}
Post hoc clustering analysis on the detected biomarkers generates biomarker modules of different meta-patterns 
that facilitate biological interpretation and provides clues for hypothesis generation and biological investigation.

Our proposed BayesMP framework has the following advantages.
Firstly, the model is simple, yet practical and powerful.
The model is based on one-sided \textit{p}-values.
This allows easy integration of data from different platforms 
(for example, many different platforms from microarray and RNA-seq).
As a contrast, 
\citet{scharpf2009bayesian} described a full Bayesian hierarchical model for microarray meta-analysis, 
where the input data are microarray raw data (normalized intensities).
Although such a full Bayesian model theoretically best integrates all information and can be more powerful,
it cannot combine new RNA-seq platforms, since RNA-seq generates count data versus continuous intensity measures in microarray.
Such a full hierarchical model also runs a greater risk of model mis-specification that increases systemic bias across different microarray platforms.
Our framework, based on  \textit{p}-values, circumvents these difficulties and is powerful as long as the method used to generate \textit{p}-values in each study is effective.
Secondly, we adopted a conjugate Bayesian approach using DPs for alternative distributions $f^{(s)}_{\pm 1}$,
which enables a mixture of multiple subgroups instead of a single one-component alternative.
DPs is nonparametric and thus robust against model assumptions. 
The conjugacy of our model guarantees the fast computing of the Gibbs sampling procedure.
Thirdly, we have shown that decision theory framework from BayesMP provides good  FDR control and power 
under different hypothesis settings (or decision spaces).
Fourthly, in contrast to the ``hard'' decision of 0 or 1 weights in AW-Fisher,
the posterior distributions of the DE indicators $Y_{gs}$ provide a stochastic quantification and ``soft'' decision.
For example, in the mouse metabolism example,
gene Mbnl2 (probeset $1422836\_at$) and gene Bcl2l11 (probeset $1435449\_at$) have very similar \textit{p}-values in the three studies:
$(0.0063, 0.16, 0.097)$ for Mbnl2 and $(0.0070, 0.16, 0.098)$ for Bcl2l11. 
Using AW,
Mbnl2 ended up with a \textit{p}-value of $0.020$ with weights $(1, 0, 1)$, 
but Bcl2l11 had a \textit{p}-value of $0.021$ with weights $(1, 1, 1)$. 
A slight alteration of the \textit{p}-value in the second study (heart tissue) resulted in different weights.
The posterior probabilities for these two genes are,
however, very similar,
with $(0.827, 0.561, 0.671)$ and $(0.819, 0.561, 0.662)$ respectively, 
and they belonged to the same gene module \Rmnum{1} in figure~\ref{fig:mouse1_modules}.
The stochastic quantification avoids sensitive 0 or 1 weight changes in AW-Fisher.
Finally,
Fisher's method does not categorize DE genes with homogeneous or heterogeneous DE patterns across studies.
In contrast, the improved AW-Fisher method allows categorization of biomarkers but generates up to  $2^S-1$ biomarker clusters, 
which becomes intractable when $S$ is large.
The posterior probability of $Y_{gs}$ from BayesMP allows the application of tight clustering to directly identify tight clusters of biomarkers with distinct DE meta-patterns.
Our simulation and three applications have shown 
good clustering accuracy and improved interpretation of the biomarker modules.

BayesMP potentially has the following potential limitations.
Computing is often a consideration for Bayesian approaches.
Our experiences have shown that 10,000 simulations are sufficient to generate the posterior probabilities in general, 
and less than 1 hour is enough for combining around $10,000$ genes using a regular desktop.
For applications to much larger numbers of  features (e.g., SNPs or methylation sites in methyl-seq),
parallel computing and/or faster MCMC techniques will be needed.

\cite{efron2004large} recommended to estimate an empirical null distribution for the \textit{Z}-statistics 
when the null distribution deviates from $\mbox{N}(0,1)$.
In our simulation, we have shown that BayesMP with theoretical null generates robust results when genes from null are correlated, 
a scenario violating the theoretical null assumption.
We also found that BayesMP with empirical null is slightly overconservative (actually FDR 1\% while nominal FDR 5\%) when noise level is large.
In our R package, we allow the users to choose from using the theoretical null or the  empirical null.
However, the user should check the assumptions made when estimating the empirical null distributions.  
For example, \cite{efron2004large} assumes less than 10\% DE genes and that empirical null is also from a Gaussian distribution.
This assumption needs to be examined post hoc (e.g., examine whether the DE proportion is less than 10\% in each study given a relatively loose FDR control, say FDR $<$ 10\%).  
To explore whether theoretical or empirical null is more appropriate,
we also provide histograms for visual diagnosis [see supplementary figures S4 and S5 \citep{supp} for details].

As mentioned in the introduction,
batch effect correction and direct merging could be a viable alternative for meta-analysis if raw data are available and batch effects can be accurately identified and corrected. 
Two major types of batch effect correction methods have been widely studied in DE analysis.
The first type considers known batch information (aka unwanted variation [UV] factors; 
e.g., experiments performed on different dates, by different technicians, or on different platforms).
Many methods have been developed to correct for these known UVs \citep[e.g.][]{johnson2007adjusting, walker2008empirical}, 
and then samples can be directly merged for so-called mega-analysis.
The second type of batch correction assumes the existence of unknown UV factors,
in which case many methods \citep[e.g.][]{leek2007capturing, kang2008accurate, listgarten2010correction} 
have been developed to eliminate effects from unknown UV factors to improve DE analysis.
Since BayesMP takes \textit{p}-values from single studies as input, 
these methods can be easily adopted in each single study before implementing BayesMP.
However, it should be noted that over-correction can be a potential concern for any batch correction method, 
and one should use with caution \citep{jacob2016correcting}.
{\color{black} In addition, such batch correction is unnecessary if each study is from a single batch and there are no other hidden factors within each study. }

A prior for $\sigma_c^2$ can be given to better characterize the variability in $\sigma_c^2$ 
(e.g., truncated inverse gamma distribution; 
here, truncation such that $\sigma_c^2 \ge 1$ is a sufficient condition such that density function of \textit{Z}-statistics is monotone with respect to \textit{Z}).
However such a prior will make the Bayesian procedure lose conjugacy.
Therefore we fix $\sigma_c^2 = 1$ to keep the algorithm computationally efficient. \citet{ghosal1999posterior} illustrated that this procedure is equivalent to choosing the bandwidth parameter a priori in kernel density estimation, and established posterior consistency for it.

BayesMP is implemented in R calling C++.
The BayesMP package is publicly available at GitHub
\url{https://github.com/Caleb-Huo/BayesMP} and the authors' websites.

\appendix

\section*{Acknowledgment}
The authors sincerely thank the editor, the associate editor, and the reviewers for their constructive comments, which helped us improve the quality of this paper.  We also thank \citet{OhioSupercomputerCenter1987} for the computational resource.

\begin{supplement}
%\sname{Supplement}
\stitle{Supplementary Information}
\slink[url]{http://url/to/supplementary.pdf}
\sdescription{Additional tables, figures, and text.}
\end{supplement}

\begin{supplement}
%\sname{}
\stitle{Supplementary Excel File 1}
\slink[url]{http://url/to/excel1.xlsx}
\sdescription{Pathway information for the mouse metabolism application.}
\end{supplement}

\begin{supplement}
%\sname{}
\stitle{Supplementary Excel File 2}
\slink[url]{http://url/to/excel2.xlsx}
\sdescription{Pathway information for the HIV transgenic rat application.}
\end{supplement}

\bibliographystyle{imsart-nameyear}
\bibliography{BayesianMeta}

\end{document}